\documentclass[11pt]{article}
\usepackage{sigsam, amsmath}
\usepackage{amsthm}
\usepackage{amssymb}
\usepackage{amsfonts}
\usepackage{mathtools}
\usepackage{maple}
\usepackage[utf8]{inputenc}
\usepackage{breqn}
\usepackage{listings}
\usepackage{moreverb}
\usepackage{verbatim}
\usepackage{sverb}
\usepackage{hyperref}
\usepackage{xcolor}
\newtheorem{theorem}{Theorem}
\newtheorem{definition}[theorem]{Definition}
\newtheorem{example}[theorem]{Example}
\providecommand{\keywords}[1]{\textbf{\textit{Keywords}} #1}
\newcommand*{\QEDA}{\hfill\ensuremath{\blacksquare}}

\issue{Vol.xx, No.xx, Issue xxx, month 202x}
\articlehead{ISSAC'22}
\titlehead{Guessing With Quadratic Differential Equations}
\authorhead{B. Teguia Tabuguia}
\begin{document}
	
	\title{Guessing With Quadratic Differential Equations}
	
	\author{Bertrand Teguia Tabuguia\\
		Nonlinear Algebra Group\\
		Max Planck Institute for Mathematics in the Sciences\\
		04103 Leipzig, Germany\\
		{\tt bertrand.teguia@mis.mpg.de}}
	
	\date{}
	
	\maketitle
	
	\begin{abstract}
		By holonomic guessing, we denote the process of finding a linear differential equation with polynomial coefficients satisfied by the generating function of a sequence, for which only a few first terms are known. Holonomic guessing has been used in computer algebra for over three decades to demonstrate the value of the guess-and-prove paradigm in intuition processes preceding proofs, as propagated in The Art of Solving (Polya, 1978). Among the prominent packages used to perform guessing, one can cite the Maple Gfun package of Salvy and Zimmermann; the Mathematica GeneratingFunctions package of Mallinger; and the Sage ore\_algebra package of Kauers, Jaroschek, and Johansson.
		
		We propose an approach that extends holonomic guessing by allowing the targeted differential equations to be of degree at most two. Consequently, it enables us to capture more generating functions than just holonomic functions. The corresponding recurrence equations are similar to known equations for the Bernoulli, Euler, and Bell numbers. As a result, our software finds the correct recurrence and differential equations for the generating functions of the up/down numbers (\url{https://oeis.org/A000111}), the evaluations of the zeta function at positive even integers, the Taylor coefficients of the Lambert W function, and many more. Our Maple implementation (\texttt{delta2guess}) is part of the FPS package which can be downloaded at \url{http://www.mathematik.uni-kassel.de/~bteguia/FPS_webpage/FPS.htm}.
	\end{abstract}
	
	
	\maketitle
	
	\section{Introduction}
	Roughly speaking, guessing is finding the solution to a given problem by over-constraining the shape of a hypothetical solution from finite data. By its nice closure properties, univariate $D$-finite functions, or simply holonomic functions, stayed at the center of this method (see \cite{kauers2022guessing,berthomieu2016guessing,mallinger1996algorithmic,zeilberger2007holonomic,gfun}). However, there are sequences and functions for which the holonomic system cannot help (see \cite{bell2008non}). For instance, the reciprocal of a holonomic function is generally not holonomic (see \cite{harris1985reciprocals}). Therefore sequences whose generating functions are reciprocals of holonomic functions are out of reach with the commonly used guessing. In \cite{BTphd,tabuguia2021representation}, the class of $\delta_2$-finite functions was introduced. We recall the definition below.
	
	\begin{definition}[$\delta_2$-finite functions]\label{def1}Assume $\frac{d^{-1}}{dz^{-1}} f = 1$ for any function $f:=f(z)$. Let $\delta_{2,z}$ be the operator defined as:
		\begin{align}\label{eq1}
	        &\delta_{2,z}^k(f) = \dfrac{d^{i-2}}{dz^{i-2}}f \cdot \dfrac{d^{j-2}}{dz^{j-2}}f = f(z)^{(i-2)}\cdot f(z)^{(j-2)},~\text{ where }~ (i,j) = \nu(k).\\
			&\nu(k) = (i,j) = \begin{cases} (l,l)~~ \text{ if }~~ N=k\\ (l+1,k-N)~\text{ otherwise}\end{cases},
		\end{align}
	
		\begin{align}
			&\text{ where } l=\left\lfloor \sqrt{2k+\frac{1}{4}}-\frac{1}{2}\right\rfloor,~\text{ and }~N=\dfrac{l(l+1)}{2}.
		\end{align}	
		A function $f(z)$ is called $\delta_2$-finite if there exist polynomials $P_0(z),\cdots,P_d(z)$ over a field $\mathbb{K}$ of characteristic zero, not all zero, such that
		\begin{equation}\label{eq2}
			P_d(z)\delta_{2,z}^{d+2}\left(f(z)\right)+\cdots+P_2(z)\delta_{2,z}^4\left(f(z)\right)+P_1(z)\delta_{2,z}^3\left(f(z)\right)+	P_0(z)f(z)=0.
		\end{equation}
	\end{definition}
	
	The derivative operator $\delta_{2,z}$ in Definition $\ref{def1}$ computes quadratic (degree 2) terms according to the lexicographic ordering in the differential polynomial ring $\mathcal{K}_f:=\mathbb{K}[z][f(z),f(z)',$ $f(z)'',f(z)^{(3)},\ldots]$. In \cite{stanley1980differentiably}, finiteness relates to the finite-dimensional vector space that spans the derivatives. In our case, this relates to the finitely many variables of $\mathcal{K}_f$ for a given $\delta_2$-finite function $f$. The author is currently working on more elaborations in this regard, including particular closure properties. The class of $\delta_2$-finite functions contains holonomic functions. It also contains reciprocals of those satisfying second-order holonomic differential equations (see \cite{tabuguia2021representation}). Functions satisfying first-order holonomic differential equations are holonomic by the arguments in \cite{harris1985reciprocals}, hence also $\delta_2$-finite. Some interesting $\delta_2$-finite functions are the Lambert W function, the reciprocals of the Bessel functions, and many generating functions involving numbers like the Bell, Bernoulli, and Euler numbers. The algorithm in \cite{tabuguia2021representation} outputs power series representations in the form
	
	\begin{equation}\label{eq3}
		f(z) = \sum_{n=0}^{\infty} a_n z^n, \begin{cases}a_{n+p} = \Phi(n,a_0,..,a_{n+p-1}),\\[3mm] a_0, a_1..,a_{p-1} \text{ taken from the series truncation of order }p\end{cases},
	\end{equation}
	where $\Phi: \mathbb{Z}\times\mathbb{K}^{n+p} \longrightarrow \mathbb{K}$ is a recursive relation deduced from the recurrence equation associated to a quadratic (or $\delta_2$-finite since it can be holonomic) differential equation of the form $(\ref{eq2})$ satisfied by $f(z)$.
	
	It is customary to find many applications of the study of sequences in enumerative combinatorics. These Bernoulli, Euler, and Bell numbers often occur in combinatorial problems. Thus, one thinks of an algorithm that will take a few terms $a_0,a_1,\ldots,a_N$ of a given sequence whose explicit formula might be unknown and find the corresponding $\Phi$ and its associated differential equation.
	
	Previously, Hebisch and Rubey (see \cite{hebisch2011extended}) investigated guessing with the whole family of algebraic differential equations, which of course, includes the $\delta_2$-finite one. However, we are not aware of a theoretical reason stipulating why one should investigate more arbitrary-degree differential polynomials than second-degree ones. We think that a systematic study of the relation between degree-k and degree-(k+1) differential polynomials will provide more insight into the use of algebraic differential equations. A simple reason to work on $\delta_2$-finite equations is that most of the algebraic differential equations in \cite{hebisch2011extended} have their $\delta_2$-finite analog -- tempting to illustrate that $\delta_2$-finite differential equations might suffice in many cases. Note, moreover, that lower-degree algebraic differential equations are usually desired when it comes to solving them using algebraic geometry techniques (see \cite{ngo2011rational2},\cite[Chapter 3]{michalek2021invitation},\cite{saito2013grobner}) or writing recursive formulas for power series solutions.
	
	\section{Sketch of the algorithm}
	
	We are given $N+1$ terms $a_0,\ldots,a_N$ of an ``unknown'' sequence $(a_n)_{n\geqslant0}$ over a field $\mathbb{K}$ of characteristic zero. We want to find a $\delta_2$-finite differential equation
	\begin{equation}\label{eq4}
		\sum_{k=0}^{d} P_k(z)\delta_{2,z}^{k+2}(f(z))=0,
	\end{equation}
	such that $f(z)=\sum_{n=0}^{\infty}a_nz^n$. The algorithm proceeds as follows:
	\begin{enumerate}
		\item Fix a degree bound $m\in\mathbb{N}$ for the polynomial coefficients $P_k(z)$, $0\leq k\leq d$. We use the default value $m=2$ since this suffices to correctly guess the equations for the Bernoulli, Euler, and Bell numbers.
		\item For $3\leq d \leq \lceil \left(N+1\right)/\left(m+1\right) \rceil$
		\begin{enumerate}
			\item Compute the recurrence relation of the power series solutions of the ansatz
			\begin{equation}
				\sum_{k=0}^{d}\left(c_{k,0}+\cdots+c_{k,m}z^m\right)\delta_{2,z}^{k+2}(f(z)),
			\end{equation}
			where $c_{k,i}, k=0\ldots d, i=0\ldots m$ are unknown coefficients.
			\item Evaluate the recurrence equation for $n=0,\ldots,(m+1)(d+1)-1$, assuming $a_n=0$ for $n\leq-1$, and using the given first terms. This yields a linear system for the unknowns $c_{k,i}$'s that we solve using Maple linear system solver (\texttt{SolveTools:-Linear}) which implements many efficient algorithms for linear systems solving. Denote by $S$ the set of solutions.
			\item If $S$ is not empty and $(m+1)(d+1)<N+1$, then use the $N+1-(m+1)(d+1)$ remaining initial terms to verify the solution found. Set $S$ to the empty set if the verification fails. Interrupt the loop (of step 2 as a whole) if the verification succeeds.
		\end{enumerate}
		\item Return the differential equation and the recurrence equation corresponding to the solution in $S$ if $S\neq \emptyset$; otherwise return FAIL (not successful).
	\end{enumerate}	
	
	\section{Examples}
	We implemented our algorithm in Maple with the name \texttt{delta2guess} (or \texttt{FPS:-delta2guess}) as a procedure of the \texttt{FPS} package available at \url{http://www.mathematik.uni-kassel.de/~bteguia/FPS_webpage/FPS.htm}. The procedure takes a list of finitely many numbers as input and returns a list of two elements in the successful case. These are a generic differential equation for the generating function and a generic recurrence equation for the input numbers, understood as the corresponding power series coefficients. 
	
	\begin{example} Maple does not, by default, compute $\zeta(50)$, where $\zeta$ denotes the Riemann Zeta function. Therefore this example is more interesting to test our implementation. Using our guessed recurrence equation, we can compute more terms of the sequence by no means of the Zeta function. The verification of the correctness is easier since $(\zeta(2n))_{n\geqslant1}$ has a well-known explicit formula in terms of the Bernoulli numbers.
		
		\begin{footnotesize}
			\begin{lstlisting}
				> L:=[seq(Zeta(2*j),j=1..24)]: L[1..5]
			\end{lstlisting}
			\begin{dmath}\label{(4)}
				\left[\frac{\pi^{2}}{6},\frac{\pi^{4}}{90},\frac{\pi^{6}}{945},\frac{\pi^{8}}{9450},\frac{\pi^{10}}{93555}\right]
			\end{dmath}
			\begin{lstlisting}
				> FPS:-delta2guess(L)
			\end{lstlisting}
			\begin{dmath}\label{(5)}
				\left[\textit{\_C} n \left(n +1\right) a \! \left(n +1\right)+\textit{\_C0} \left(n -1\right) n a \! \left(n \right)-2 \textit{\_C} \left(\moverset{n -1}{\munderset{k =0}{\textcolor{gray}{\sum}}}\! \left(k +1\right) a \! \left(k +1\right) a \! \left(n -1-k \right)\right)-2 \textit{\_C0} \left(\moverset{n -2}{\munderset{k =0}{\textcolor{gray}{\sum}}}\! \left(k +1\right) a \! \left(k +1\right) a \! \left(n -2-k \right)\right)-\textit{\_C} \left(\moverset{n}{\munderset{k =0}{\textcolor{gray}{\sum}}}\! a \! \left(k \right) a \! \left(n -k \right)\right)-\textit{\_C0} \left(\moverset{n -1}{\munderset{k =0}{\textcolor{gray}{\sum}}}\! a \! \left(k \right) a \! \left(n -1-k \right)\right)+\frac{5 n \textit{\_C0} a \! \left(n \right)}{2}+\frac{5 \textit{\_C} \left(n +1\right) a \! \left(n +1\right)}{2} \hiderel{=} 0,
				\\
				\left(-\textit{\_C0} z -\textit{\_C} \right) y \! \left(z \right)^{2}+\left(-2 \textit{\_C0} \,z^{2}-2 \textit{\_C} z \right) y \! \left(z \right) \left(\frac{d}{d z}y \! \left(z \right)\right)+\left(\textit{\_C0} \,z^{2}+\textit{\_C} z \right) \left(\frac{d^{2}}{d z^{2}}y \! \left(z \right)\right)+\left(\frac{5 \textit{\_C}}{2}+\frac{5 \textit{\_C0} z}{2}\right) \left(\frac{d}{d z}y \! \left(z \right)\right)\hiderel{=}0\right]
			\end{dmath}
		\end{footnotesize}
		
		Substituting $\_C=2$ and $\_C0=0$ gives the expected guess. That is equivalent to taking an element of the basis of the $\mathbb{K}$-module of all solutions. \QEDA
		\end{example} 
		
		Therefore we obtain the following theorem.
		
		\begin{theorem}
			The generating function of the sequence defined by
			\[a(n):=\zeta(2n+2)=(-1)^{n}(2\pi)^{2n+2}\frac{B_{2n+2}}{2\,(2n+2)!}, ~n\geq 0 \]
			where $B_n$ is the $n^{\text{th}}$ Bernoulli number, satisfies the recurrence equation
			\begin{equation}\label{eq5}
				\left(2 n + 5\right) \left(n +1\right) a \! \left(n +1\right) - \! 2 \left(\sum_{k=0}^{n}\! 2 \! \left(k +1\right) a \! \left(k +1\right) a \! \left(n -1-k \right) + \! a \! \left(k \right) a \! \left(n -k \right)\right) = 0,
			\end{equation}
			and, equivalently, its generating function is a solution of the differential equation
			\begin{equation}\label{eq6}
				2 z \left(\frac{d^{2}}{d z^{2}}y \! \left(z \right)\right)+5 \frac{d}{d z}y \! \left(z \right)-4 z y \! \left(z \right) \left(\frac{d}{d z}y \! \left(z \right)\right)-2 y \! \left(z \right)^{2}=0.
			\end{equation}
		\end{theorem}
		\begin{proof}
			The generating function of $(a(n))_{n\geqslant0}$ can be deduced from the well-known relation
			\begin{equation}\label{eq7}
				-\frac{\pi\,z}{2}\cot(\pi\,z)+\frac{1}{2} = \sum_{n=1}^{\infty}\zeta(2n) z^{2n}.
			\end{equation}
			Therefore we get the generating function $\frac{1-\pi \cot{\left( \pi \sqrt{z}\right) } \sqrt{z}}{2 z}$ which satisfies $(\ref{eq6})$.
		\end{proof}
	
	\begin{example} The up/down or ``zig-zag'' numbers from \url{https://oeis.org/A000111} have the exponential generating function $\tan\left(\frac{z}{2}+\frac{\pi}{4}\right)$. This can be recovered by our implementation. Indeed, the values from OEIS, each divided by the factorial of their index (starting from zero), lead to the following guess.
		\begin{dmath}\label{(20)}
			\left[\left(n +1\right) \left(n +2\right) a \! \left(n +2\right)-\left(\moverset{n}{\munderset{k =0}{\textcolor{gray}{\sum}}}\! \left(k +1\right) a \! \left(k +1\right) a \! \left(n -k \right)\right) \hiderel{=} 0,
			\frac{d^{2}}{d z^{2}}y \! \left(z \right)-y \! \left(z \right) \left(\frac{d}{d z}y \! \left(z \right)\right)\hiderel{=}0\right]
		\end{dmath}
		Solving the differential equation and matching the first two Taylor coefficients with the first two terms of the sequence give exactly the desired generating function. Note that using \texttt{FPS:-FindQRE} and \texttt{FPS:-QDE}, one can also compute the same recurrence equation and differential equation from $\tan\left(\frac{z}{2}+\frac{\pi}{4}\right)$, respectively.  \QEDA
	\end{example}

\end{document}